\documentclass[reqno,11pt]{amsart}

\usepackage[a4paper,  margin=3.6cm]{geometry}

\usepackage{amsmath,amssymb}
\usepackage{dsfont}         
\usepackage{comment}
\newtheorem{theorem}{Theorem}[section]

\newtheorem{lemma}[theorem]{Lemma}
\newtheorem{proposition}[theorem]{Proposition}
\theoremstyle{definition}
\newtheorem{definition}[theorem]{Definition}

\theoremstyle{remark}
\newtheorem{remark}[theorem]{Remark}
\numberwithin{equation}{section}

\title[Uniqueness of Gibbs fields on unbounded degree graphs]{Uniqueness of Gibbs fields with
unbounded random interactions on unbounded degree graphs}

\author{Dorota K\c{e}pa-Maksymowicz}

\address{Instytut Matematyki, Uniwersytet Marii Curie-Sk{\l}odowskiej, 20-031 Lublin, Poland}
\email{dkm@umcs.lublin.pl}

\author{ Yuri Kozitsky}

\address{Instytut Matematyki, Uniwersytet Marii Curie-Sk{\l}odowskiej, 20-031 Lublin, Poland}
\email{jkozi@hektor.umcs.lublin.pl}


\begin{document}
\maketitle

\begin{abstract} Gibbs fields with continuous spins are studied, the underlying
graphs of which can be of unbounded vertex degree and the spin-spin
pair interaction potentials are random and unbounded. A
high-temperature uniqueness of such fields is proved to hold under
the following conditions: (a) the vertex degree is of tempered
growth, i.e., controlled in a certain way; (b) the interaction
potentials $W_{xy}$ are such that $\|W_{xy}\|=\sup_{\sigma,\sigma'}
|W_{xy}(\sigma, \sigma')|$ are independent (for different edges
$\langle x, y \rangle$), identically distributed and exponentially
integrable random variables.

\end{abstract}

\section{Introduction and Setup}

In this work, we continue studying quenched Gibbs fields with
unbounded disorder \cite{KKP} focusing on their {\it
high-temperature} uniqueness. Permanent interest to this problem may
arise from the fact that -- even in the simplest case of an Ising
model with unbounded random interactions -- due to so-called
Griffiths' singularities \cite{DKP,F,SY} at arbitrarily high
temperatures there may exist arbitrarily large subsets of the
underlying lattice, in which the spins are strongly correlated. Our
work can be considered as a continuation of the research performed
in \cite{BD,Berg,DKP,GM,KK1}. Novel aspects here are: (a) instead of
finite-valued spins  we allow them to take values in arbitrary
Polish spaces; (b) instead of regular underlying graphs (like
$\mathds{Z}^d$) we employ graphs of unbounded vertex degree, cf.
\cite{KK2,KKP0}. Markov random fields on such underlying graphs
\cite{Hagg} naturally appear in the following physical applications:
(a) random (also quantum) fields on Riemannian manifolds, see e.g.,
\cite{deA}; (b) thermodynamic states of interacting oscillators
based on networks -- so called oscillating networks \cite{Bur}; (c)
thermodynamic states of continuous systems with spins, like
ferrofluids, where (random) geometric graphs are used as underlying
graphs \cite{DKKP,RZ}. We refer the interested reader to \cite{KKP0}
for more details on this matter.

Let ${\sf G}=({\sf V},{\sf E})$ be a countably infinite graph with
vertex and edge sets ${\sf V}$ and ${\sf E}$, respectively. We
assume that ${\sf G}$ is connected, has no loops and multiple edges,
and
\begin{equation}
  \label{Graph}
\forall x \in {\sf V} \qquad \ n(x) := \#\{ y\in {\sf V}: y \sim x\}
<\infty,
\end{equation}
where $x\sim y$ denotes adjacency. In the latter case, we write the
corresponding edge as $\langle x, y \rangle$. Let $S$ be a Polish
space -- a complete and separable metric space. By $\mathcal{B}(S)$
we denote the corresponding Borel $\sigma$-field, whereas
$\mathcal{F}$ will stand for the smallest $\sigma$-field of subsets
of $S^{\sf V}$ such that the maps $S^{\sf V}\ni \sigma \mapsto
\sigma (x) \in S$ are measurable for all $x\in {\sf V}$. By
$\mathcal{P}(S^{\sf V})$ we denote the set of all probability
measures on $(S^{\sf V}, \mathcal{F})$. For a $\mathit{\Delta}
\subset {\sf V}$, we write $\mathit{\Delta}^c = \ {\sf V} \setminus
\mathit{\Delta}$; $\mathit{\Delta} \Subset {\sf V}$ means
$\mathit{\Delta} \subset {\sf V}$ and $\mathit{\Delta}$ is non-void
and finite. For $\mathit{\Delta}\subset {\sf V}$,
$\mathcal{F}_{\mathit{\Delta}}$ stands for the sub-field of
$\mathcal{F}$ generated by the maps $S^{\sf V}\ni \sigma \mapsto
\sigma (x)\in S$ with $x\in {\mathit{\Delta}}$. For
$\mathit{\Delta}\Subset {\sf V}$, a {\it probability kernel}
$\pi_{\mathit{\Delta}} (\cdot| \cdot)$ is a function on
$(\mathcal{F}, S^{\sf V})$ such that for any $\xi \in S^{\sf V}$,
$\pi_{\mathit{\Delta}}(\cdot | \xi)$ is in $\mathcal{P}(S^{\sf V})$,
and for any $A \in \mathcal{F}$, $\pi_{\mathit{\Delta}}(A | \cdot )$
is $\mathcal{F}_{\mathit{\Delta}^c}$-measurable. Such a kernel is
said to be proper if $\pi_{\mathit{\Delta}} (A| \cdot) =
\mathbb{I}_A(\cdot)$ for any $A \in
\mathcal{F}_{\mathit{\Delta}^c}$, cf. \cite[page 14]{Ge}. Here
$\mathbb{I}_A(\xi) = 1$ if $\xi \in A$, and $\mathbb{I}_A(\xi) = 0$
otherwise. Let $\{\pi_{\mathit{\Delta}}\}_{\mathit{\Delta}\Subset
{\sf V}}$ be the family of probability kernels such that for a given
$\mu \in \mathcal{P}(S^{\sf V})$, one has
\begin{equation}
 \label{i1}
\mu(A |\mathcal{F}_{\mathit{\Delta}^c}) = \pi_{\mathit{\Delta}} (A|
\cdot),
\end{equation}
which holds $\mu$-almost surely for all $A\in \mathcal{F}$ and
$\mathit{\Delta} \Subset {\sf V}$. Then one says \cite[page 16]{Ge}
that $\mu$ is {\it specified} by
$\{\pi_{\mathit{\Delta}}\}_{\mathit{\Delta}\Subset {\sf V}}$. In
this case, all the kernels $\pi_{\mathit{\Delta}}$ are $\mu$-almost
surely proper, and their family is $\mu$-almost surely consistent.
The latter means that for $\mu$-almost all $\xi$ and all $A\in
\mathcal{F}$, one has
\begin{equation*}
 \label{i2}
\int_{S^{\sf V}} \pi_{\mathit{\Lambda}} (A| \eta)
\pi_{\mathit{\Delta}} (d \eta | \xi) = \pi_{\mathit{\Delta}} (A|
\xi),
\end{equation*}
which holds for any pair of subsets such that $\mathit{\Lambda}
\subset \mathit{\Delta}$. It should be pointed out that (\ref{i1})
holds if and only if
\begin{equation}
  \label{i0}
  \int_{S^{\sf V}} \pi_{\mathit{\Delta}} (A|\xi) \mu(d \xi) = \mu(A),
\end{equation}
for all $A\in \mathcal{F}$ and $\mathit{\Delta}\Subset {\sf V}$. The
condition (\ref{i0}) can be considered as the equation which defines
the random fields specified by the family of kernels
$\{\pi_{\mathit{\Delta}}\}_{\mathit{\Delta}\Subset {\sf V}}$. It is
known under the name Dobrushin-Lanford-Ruelle (DLR) equation. A
paramount problem of the theory is then the number of such fields
specified by a given family
$\{\pi_{\mathit{\Delta}}\}_{\mathit{\Delta}\Subset {\sf V}}$, see
\cite{Ge}.

Let $W= \{W_{xy}\}_{\langle x,y \rangle\in {\sf E}}$  be a family of
symmetric measurable functions $W_{xy}: S \times S \to \mathds{R}$,
interpreted as the interaction potentials, i.e., $W_{xy}(\sigma (x),
\sigma (y))\in \mathds{R}$ is the energy related to the spin values
$\sigma (x), \sigma (y) \in S$. For $\mathit{\Delta} \Subset {\sf
V}$ and $\xi\in S^{\sf V}$, we set
\begin{equation}
 \label{2}
H_{\mathit{\Delta}} (\sigma|\xi) = -\frac{1}{2} \sum_{x,y\in
\mathit{\Delta} , \ x \sim y} W_{xy}(\sigma(x) , \sigma(y)) -
\sum_{x\in \mathit{\Delta} , \ y \in \mathit{\Delta}^c, \ x \sim y}
W_{xy}(\sigma(x) , \xi(y)).
\end{equation}
Let $\chi = \{\chi_x\}_{x\in {\sf V}}$ be a family of $\chi_x \in
\mathcal{P}(S)$. For $\varDelta \subset {\sf V}$, by
$\chi_{\varDelta}$ we denote the corresponding product measure on
$S^\varDelta$. The elements of the latter set are denoted by
$\sigma_{\varDelta}$. The local Gibbs specification corresponding to
(\ref{2}) and $\chi$ is the collection of the following probability
kernels
\begin{eqnarray}
 \label{3}
 \pi^\beta_{\mathit{\Delta}} (A|\xi) & = & \frac{1}{Z^\beta_{\mathit{\Delta}}(\xi)}\int_{S^{\mathit{\Delta}}}
 \mathbb{I}_A (\sigma_{\mathit{\Delta}}\times \xi_{\mathit{\Delta}^c})\\[.2cm]
 &\times& \exp\left[ - \beta H_{\mathit{\Delta}} (\sigma|\xi)\right] \chi_{\mathit{\Delta}}
 (d\sigma_{\mathit{\Delta}}), \quad A\in \mathcal{F}, \quad \beta>0. \nonumber
\end{eqnarray}
Here $\beta$ is the inverse absolute temperature,
$Z^\beta_{\mathit{\Delta}} (\xi)$ is the normalizing factor -- {\it
partition function}, and the juxtaposition stands for the element of
$S^{\sf V}$ such that
\[
(\sigma_{\mathit{\Delta}} \times \xi_{\mathit{\Delta}^c} ) (x) =
\sigma (x), \  {\rm for} \ x\in \mathit{\Delta}; \quad \
(\sigma_{\mathit{\Delta}} \times \xi_{\mathit{\Delta}^c} ) (x) = \xi
(x), \  {\rm for} \ x\in \mathit{\Delta}^c.
\]
Then by $\mathcal{G}^\beta (W, \chi)$ we denote the set of $\mu \in
\mathcal{P}(S^{\sf V})$ that solve the DLR equation (\ref{i0}) with
the specification defined in (\ref{3}).

As mentioned above, we study the case where the interaction
potentials $W_{xy}$ are random functions, and hence
$\mathcal{G}^\beta (W, \chi)$ is a random set, see \cite{KKP} for a
more detailed presentation of the corresponding notions. In the
Ising model studied in \cite{DKP,GM}, $W_{xy}(\sigma (x), \sigma
(y))= J_{xy} \sigma (x) \sigma (y)$ with random $J_{xy}$. The aim of
this work is to find a sufficient condition under which
$\mathcal{G}^\beta (W, \chi)$ is almost surely a singleton.

\section{Tempered Graphs and the Statement}

We begin by introducing the underlying graph $\sf G$ and the
corresponding graph-related terminology and facts. Afterwards, we
formulate and comment our result.

\subsection{Tempered graphs}

As mentioned above, the underlying graph ${\sf G}=({\sf V},{\sf E})$
is supposed to satisfy (\ref{Graph}), which means that is is
\emph{locally finite}. At the same time, we allow ${\sf G}$ to be
such that
\begin{equation}
  \label{Gra}
  \sup_{x\in {\sf V}} n(x) = +\infty.
\end{equation}
Let ${\sf G'} = ({\sf V}', {\sf E}')$ be a subgraph of $\sf G$,
i.e., ${\sf G}'$ is a graph and ${\sf V}'\subset {\sf V}$, ${\sf E}'
\subset {\sf E}$. In this case, we write ${\sf G}' \subset {\sf G}$.
A finite connected subgraph of $\sf G$ is called an {\it animal}.
Such zoo-terminology is used in the theory of Gibbs states based on
cluster expansions, see, e.g., \cite{DobS}, and in
percolation-related works, e.g., \cite{Cox}. For an animal ${\sf
A}$, by ${\sf V}({\sf A})$ and ${\sf E}({\sf A})$ we denote its
vertex and edge sets, respectively. For graphs of bounded vertex
degree, i.e., those satisfying
\begin{equation}
  \label{GraA}
  \sup_{x\in {\sf V}} n(x) =:\bar{n} < +\infty,
\end{equation}
one has
\begin{equation}
  \label{GraB}
G({\sf A}; g):=\frac{1}{|{\sf V}({\sf A})|} \sum_{x \in {\sf V}({\sf
A})} g\left(n(x) \right) \leq g(\bar{n}),
\end{equation}
holding for each unbounded $g:[1, +\infty)\to [0, +\infty)$. Then a
possible way of controlling deviations from the boundedness as in
(\ref{GraA}) is to impose weaker versions of (\ref{GraB}), which
would make impossible accumulations of \emph{hubs} -- vertices of
high degree. This idea was realized in our previous work \cite{KK2}
the results of which will be used here. Since we are going to modify
 the basic notion of \cite{KK2} -- and also for reader's convenience --
we begin by presenting facts and notions related to this matter.

A path in ${\sf G}$ is a finite sequence of its vertices
$\vartheta=\{x_l: l=0,1, \dots, n\}$ such that $x_l \sim x_{l+1}$
for all $l$. Vertices $x_0$ and $x_n$ are its origin and terminus,
respectively, and $n=:\|\vartheta\|$ is the length of the path. It
is equal to the number of consecutive pairs $x_l, x_{l+1}$ in
$\vartheta$. By writing $\vartheta(x,y)$ we indicate that the path
originates at $x$ and terminates at $y$, and thereby connects these
vertices. For consecutive $x_l$, $x_{l+1}$, we say that the path
\emph{leaves} $x_l$ and enters $x_{l+1}$. Some of the vertices can
appear in $\vartheta$ several times. By $\upsilon_\vartheta (x)$,
$x\in {\sf V}$, we denote the number of times $\vartheta$ leaves
$x$, including $\upsilon_\vartheta (x)=0$ if $x$ is not in
$\vartheta$. A subsequence $\vartheta'$ of $\vartheta$ is a
\emph{subpath} of the latter if $\vartheta'$ is a path on its own. A
path $\vartheta$ is called \emph{simple} if $\upsilon_\vartheta (x)
\leq 1$ for all $x$. A simple path with coinciding origin and
terminus is called a cycle. In a simple path, each vertex appears
only once -- except, possibly, for the origin and terminus which can
coincide if $\vartheta$ is a cycle. Each path $\vartheta(x,y)$
contains a subpath $\vartheta'(x,y)$, which is simple.

For a path $\vartheta$, let ${\sf G}_\vartheta=({\sf V}_\vartheta,
{\sf E}_\vartheta)$ be its graph. Here ${\sf V}_\vartheta$ and ${\sf
E}_\vartheta$ consist of the vertices $x_l$ and edges $\langle x_l,
x_{l+1}\rangle$, respectively, where each of them is counted only
once. Clearly, ${\sf G}_\vartheta=({\sf V}_\vartheta, {\sf
E}_\vartheta)$ is an animal since the connectivity is defined as the
existence of a path connecting each pair of distinct vertices.  By
Euler's classical theorem, see, e.g., \cite[page 51]{BM} it is
possible to show that each ${\sf A}$ is in fact ${\sf G}_\vartheta$
for a certain path $\vartheta$ such that $\upsilon_\vartheta (x)
\leq n (x)$, see the proof of Lemma 17 in \cite{KK2} for more
detail. For $x,y\in {\sf V}$, the distance between them is defied as
\begin{equation}
  \label{GraC}
  \rho (x,y) = \min \|\vartheta(x,y)\|,
\end{equation}
where min is taken over all paths connecting $x$ and $y$. With this
distance $\sf G$ is a metric space.

Now let $\sf G$ be such that (\ref{Gra}) holds, and $g$ be an
increasing ad infinitum function. Then $G({\sf A};g)$ defined in
(\ref{GraB}) can be made arbitrarily big just by picking a hub and
then including in $\sf A$ this hub and a couple of its neighbors.
Loosely speaking, a graph is declared tempered if hubs are sparse,
i.e., most of them are at large distances (\ref{GraC}) of each
other. This means that big animals can contain only few hubs, and
hence $G({\sf A};g)$ for such animals can be bounded, uniformly for
all big $\sf A$. Let us make this idea more precise. For given $x\in
{\sf V}$ and $r\in \mathds{N}$, let ${\sf G}_r(x)$ be the subgraph
with vertex set ${\sf V}_r(x) = \{ y\in {\sf V}: \rho(x,y) \leq r\}$
and edge set consisting of all those edges of ${\sf G}$ both
endpoints of which are in ${\sf V}_r(x)$. Now we set
\begin{equation}
  \label{GraD}
\mathcal{A}_r(x) =\{{\sf A}\subset {\sf G}_r(x) : |{\sf V}({\sf A})|
\geq r+1\}.
\end{equation}
In the definition below, by $g$ we mean any increasing ad infinitum
function $g:[1+\infty)\to [0,+\infty)$.
\begin{definition}
\label{KKdf} ${\sf G}$ is said to be $g$-tempered if for every $x
\in {\sf V}$, there exists a strictly increasing sequence
$\{N_k\}_{k\in \mathds{N}} \subset \mathds{N}$ such that
\begin{equation}
  \label{110}
  \sup_{x\in {\sf V}}\sup_{k\in \mathds{N}}\bigg{(}\max_{{\sf A}\in \mathcal{A}_{N_k}(x)} G\left({\sf A};g \right)
  \bigg{)}=:\gamma < \infty.
\end{equation}
\end{definition}
The following rooted tree, cf. \cite[Section 4]{KK1},  can serve as
a natural example characterizing the property just defined. In this
tree, the root $x$ has two neighbors, which have three further
neighbors each (not counting $x$). These neighbors have four further
neighbors each, and so on. It is an easy exercise to show that such
a graph is not $g$-tempered for any unbounded $g$. Notably, in this
graph each hub has an even bigger hub at distance one.

Another natural example characterizing the property defined above is
provided by so called repulsive graphs, see \cite{KK2}. Some of such
graphs were introduced in \cite[Theorem 3]{BD}. To describe them we
set $m_{-}(x,y) = \min\{n(x); n(y)\}$, $x,y\in  {\sf V}$, and let
$\phi:[1,+\infty) \to [0,+\infty)$ be an increasing ad infinitum
function. Then the set of graphs $\mathbb{G}_{-}(\phi)$ is defined
by the following property: for each ${\sf G} \in
\mathbb{G}_{-}(\phi)$, there exists $n_*\in \mathds{N}$ such that,
for each $x,y\in {\sf V}$, the following holds
\begin{equation}
  \label{GraE}
  \rho (x,y) \geq \phi(m_{-}(x,y)),
\end{equation}
whenever $m_{-} (x,y) \geq n_*$.
\begin{proposition}
  \label{Gua1pn}
  Assume that there exists a strictly increasing sequence
  $\{t_k\}_{k\in \mathds{N}}$, $t_k \to +\infty$, such that $g$ and
  $\phi$ satisfy
  \begin{equation}
    \label{GraF}
 \sum_{k=1}^{\infty} \frac{g(t_{k+1})}{\phi (t_k)} < +\infty.
  \end{equation}
Then each ${\sf G}\in \mathbb{G}_{-}(\phi)$ is $g$-tempered.
\end{proposition}
This statement resembles Theorem 5 of \cite{KK2}. We give its proof
in this article below  since the definition of temperedness used
here is a bit different from that made in \cite[Definition 1]{KK2}.
To this end we will need some facts proved in \cite{KK2}.
\begin{proposition}{\rm \cite[Lemma 17]{KK2}}
  \label{Gua2pn}
Let ${\sf G}$ be an arbitrary graph. Given $\lambda >1$ and an
animal ${\sf A}\subset {\sf G}$, let ${\sf B}\subset {\sf V}({\sf
A})$ be such that $\rho (x,y) \geq \lambda$ for each pair of
distinct $x,y\in {\sf B}$. Then
\begin{equation}
\label{GG} |{\sf B}| \leq \max\left\{1; \frac{ 2 |{\sf V}({\sf A})|
-1}{\lambda} \right\}.
\end{equation}
\end{proposition}
Before formulating the next result we recall that $\phi$ is an
increasing ad infinitum function and the ball ${\sf V}_r(x)=\{y:
\rho(x,y)\leq r\}$ is the set of vertices of ${\sf G}_r(x)$, see
(\ref{GraD}).
\begin{proposition}{\rm \cite[Lemma 18]{KK2}}
  \label{Gua3pn}
Let ${\sf G}$ be in $\mathbb{G}_{-}(\phi)$. Then for each $x\in {\sf
V}$, there exists a strictly increasing sequence $\{N_k\}_{k\in
\mathds{N}}\subset \mathds{N}$ such that, for each $k\in
\mathds{N}$, the following holds
\begin{equation}
  \label{GG1}
  \max_{y\in {\sf V}_{N_k} (x)} n(y) \leq \phi^{-1} ( 2 N_k +1).
\end{equation}
\end{proposition}
{\it Proof of Proposition \ref{Gua1pn}:} For a given $x\in {\sf V}$,
let $\{N_k\}_{k\in \mathds{N}}$ be as in Proposition \ref{Gua3pn}.
For ${\sf A}\in \mathcal{A}_{N_k} (x)$, see (\ref{GraD}), by
(\ref{GG1}) we then have $\phi(n_{\sf A}) \leq 2 |{\sf V}({\sf
A})|$, $n_{\sf A}:= \max_{y\in {\sf V}({\sf A})} n(y)$. Now let
$\{t_k\}_{k\in \mathds{N}}$ be such that (\ref{GraF}) holds. Set
${\sf M}_k = \{ y \in {\sf V}({\sf A}): n(y) \in [t_k, t_{k+1})\}$,
$k\in \mathds{N}$. By (\ref{GraE}) it follows that $\rho( y, y')
\geq \phi (t_k)$ whenever $y,y' \in {\sf M}_k$. By (\ref{GG}) this
yields
\[
|{\sf M}_k| \leq \frac{2 |{\sf V}({\sf A})|}{\phi(t_k)}.
\]
With the help of this estimate we then have
\begin{gather*}
  \sum_{y\in {\sf V}({\sf A})} g(n(y)) \leq \sum_{k=1}^\infty
  g(t_{k+1}) |{\sf M}_k| \leq 2 \left( \sum_{k=1}^\infty \frac{g(t_{k+1})}{\phi(t_k)}\right)  |{\sf V}({\sf
  A})| =: \gamma |{\sf V}({\sf
  A})| ,
\end{gather*}
which completes the proof.

Now let $\varSigma_r (x)$ be the set of all simple paths $\vartheta$
(excluding cycles) such that ${\sf G}_\vartheta \subset {\sf
G}_r(x)$ and $\|\vartheta\|\geq r$; hence, $|{\sf V}_\vartheta|\geq
r+1$. Clearly, ${\sf G}_\vartheta \in \mathcal{A}_r(x)$ for each
$\vartheta\in \varSigma_r (x)$, see (\ref{GraD}). Set $g_1(t) = \log
t$, $g_2 (t) = t\log t$, $t\geq 1$, and also $\varSigma_r^N (x) = \{
\vartheta\in \varSigma_r(x) : \|\vartheta\|= N\}$, $\mathcal{A}_r^N
(x) = \{{\sf A} \in \mathcal{A}_r(x): |{\sf V}({\sf A})|=N\}$. The
next statement provides estimates of the cardinalities of these
sets.
\begin{proposition}
  \label{GGpn}
Let ${\sf G}$ be $g_1$-tempered, $\gamma>0$ be as in (\ref{110}),
$x\in {\sf V}$ and $\{N_k\}_{k\in \mathds{N}}$ be  as in Definition
\ref{KKdf}. Then for each $k\in \mathds{N}$ and $N\geq N_k$, the
following holds $|\varSigma_{N_k}^N (x)| \leq \exp( \gamma N)$. If
${\sf G}$ is $g_2$-tempered, then $|\mathcal{A}_r^N (x)|  \leq \exp(
\gamma N)$ holding for all $N\geq N_k+1$.
\end{proposition}
\begin{proof}
Let $\varTheta$ be a finite family of paths in a graph ${\sf G}$.
Proceeding as in the proof of \cite[Lemma 15]{KK2} we obtain the
following estimate of its cardinality
\begin{gather}
  \label{GG3}
|\varTheta| \leq \max_{\vartheta \in \varTheta} \exp\left(
\sum_{y\in \vartheta} \log n(y)\right) = \max_{\vartheta \in
\varTheta} \exp\left( \sum_{y\in {\sf
V}_\vartheta}\upsilon_\vartheta (y) \log n(y)\right),
\end{gather}
where $\upsilon_\vartheta (y)$ is the number of times $\vartheta$
leaves $y$. For $\vartheta\in \varSigma_{N_k}^N (x)$, we have
$\upsilon_\vartheta (y) \leq 1$, which by (\ref{GG3}) and
(\ref{110}) yields the estimate in question. As mentioned above, for
each ${\sf A}$, there exists $\vartheta$ such that ${\sf A}= {\sf
G}_\vartheta$ and $\upsilon_\vartheta (y) \leq n(y)$ for each $y\in
{\sf V}_\vartheta = {\sf V}({\sf A})$. Let $\varTheta_{r}^N (x)$ be
the family of all paths in ${\sf G}_r(x)$ such that $|{\sf
G}_\vartheta|=N$ and $\upsilon_\vartheta (y) \leq n(y)$ for all
$y\in {\sf V}_\vartheta$. Then $|\mathcal{A}_{N_k}^N (x)| \leq
|\varTheta_{N_k}^N (x)| \leq \exp( \gamma N)$, where the latter
estimate is obtained in the same way as the estimate for
$|\varSigma_{N_k}^N (x)|$. This completes the proof.
\end{proof}

Now let us make some comments.
\begin{itemize}

\item[(i)] As mentioned above, in graphs with unbounded degrees
$G({\sf A};g)$ can be arbitrarily big if one takes ${\sf A}$ `small'
and containing a hub. According to
 Definition \ref{KKdf}, in a tempered graph there exists a scale of `big enough' animals, for
 which $G({\sf A};g)\leq \gamma$ with a universal (for this ${\sf
 G}$) constant $\gamma$. This scale is defined by means of balls -- typical objects for metric spaces.
 Temperedness introduced in this work is a
 bit stronger than the corresponding notion established in
 \cite[Definition 1]{KK2}. The main difference is that in the
 present  case, the condition $G({\sf A};g)\leq \gamma$ is satisfied
 by \emph{all} animals ${\sf A}\subset {\sf G}_{N_k}(x)$ -- not only
 those for which $x\in {\sf V}({\sf A})$ and $|{\sf V}({\sf A})| =
 N_k$, as is assumed in \cite{KK2}.
  \item[(ii)]
Proposition \ref{Gua1pn} provides a basic example of a tempered
graph. It is, however, not the only one satisfying Definition
\ref{KKdf}. For ${\sf G}\in \mathbb{G}_{-}(\phi)$, if $n(x)\geq
n_*$, then the ball ${\sf B}_x:=\{y: \rho(x,y) \leq \phi(n(x))-1\}$
contains no $z\in {\sf V}$ such that $n(z) \geq n(x)$. It is clear
that the graph remains tempered if one allows such balls contain up
to a fixed number of hubs $z$ for which $n(z) = n(x)$. In other
words, our tempered graphs include repulsive graphs as a particular
case.
//\item[(iii)]
 In
\cite[Assumption 2.1]{KKP0}, there was used a condition based on the
weighted summability of $$m_\theta (x) = \sum_{y: y\sim x} [n(x)
n(y)]^\theta, \quad \theta >0.$$ In mathematical chemistry,
$m_\theta$ is known as generalized Randi\'c index, see \cite{CG}. It
turns out, see \cite[Theorem 5.2]{KKP0}, that graphs from
$\mathbb{G}_{+}(\phi)$ satisfy this condition for a specific choice
of $\phi$. Here the family $\mathbb{G}_{+}(\phi)$ is defined by the
repulsive condition as in (\ref{GraE}) with $m_{-}$ replaced by
$m_{+} (x,y) = \max\{n(x); n(y)\}$. More on the properties of
$\mathbb{G}_{+}(\phi)$ and Randi\'c indices in this context can be
found in \cite{KK2}.
\end{itemize}

\subsection{The statement}

Let $\mathcal{C}$ be the Banach space of symmetric bounded
continuous functions $S\times S \ni (s,s') \mapsto f(s,s')\in
\mathds{R}$, equipped with the norm
\begin{equation*}
  \label{4}
  \|f\| = \sup_{(s,s')\in S\times S} |f(s,s')|,
\end{equation*}
and with the corresponding Borel $\sigma$-field. We then equip
$\mathcal{C}^{\sf E}$ with the product topology and the product
$\sigma$-field of subsets. Let $(\Omega, \mathcal{O},P)$ be a
probability space. Then a random interaction potential is a
measurable map $W:\Omega \to \mathcal{C}^{\sf E}$. In this work, it
is supposed to satisfy:
\begin{eqnarray}\label{zalgl}
& & \textrm{(i) } \forall \langle x, y\rangle \in {\sf E} \quad \|W_{xy}\|=\sup_{\sigma(x),\sigma (y)\in S}|W_{xy}(\sigma (x),\sigma (y))|<\infty,\\
\nonumber & & \textrm{(ii) } \left\{\|W_{xy}\|,\; \langle x,y\rangle\in {\sf E}\right\} \textrm{ are i.i.d.},\\[.2cm]
\nonumber & & \textrm{(iii) } \forall \langle x,y\rangle\in {\sf
E}\quad \forall \beta>0\quad {\mathbb
E}\left[\exp(\beta\|W_{xy}\|)\right]<\infty.
\end{eqnarray}
By (\ref{Graph}) and property (i) above the interaction potential is
almost surely absolutely summable, cf. \cite[Subsec. 2.11, pages 28,
29]{Ge}. Then, for all $\beta >0$, by \cite[Theorem 4.23 claim (a),
page 72]{Ge}, $\mathcal{G}^\beta (W,\chi)$ is almost surely
non-void. The result of this paper is the following statement.
\begin{theorem}
  \label{2tm}
Let ${\sf G}$ be $g$-tempered with $g(t)=\log t$, cf. Proposition
\ref{GGpn}, and the conditions in (\ref{zalgl}) be satisfied. Then
there exist $\beta_*>0$ such that, for all $\beta<\beta_*$, the set
$\mathcal{G}^\beta(W,\chi)$ is almost surely a singleton.
\end{theorem}
Let us make some comments to this statement.
\begin{itemize}
  \item[(a)] Up to the  best of our knowledge, this is the first result of this
kind for unbounded degree graphs. The temperedness of the underlying
graph as in Definition \ref{KKdf} is crucial for the uniqueness of
Gibbs fields constructed thereon -- even in the simplest case of
nonrandom interactions of spins taking values $\pm1$. The Ising
model on the rooted tree mentioned above has multiple phases at all
temperatures, see \cite[Section 4]{KK1}.
\item[(b)]
 Perhaps, the first rigorous result
concerning uniqueness of Gibbs fields with regular underlying graphs
(in fact, lattices), finite-valued spins and unbounded random
interactions was obtained in \cite{BD}. Despite the title of that
paper, the technique used there -- `gluing out' vertices of the
original lattice with subsequent passing to coarse-grained graphs of
a special structure -- was quite complicated and not everywhere
correct. For example, the proof of Proposition 3 was based on a
wrong estimate of the distance between vertices of the
coarse-grained graph. Fortunately, this gap can be bridged and thus
the main statement of \cite{BD} survives. Our result -- essentially
more general -- is obtained in a much more natural and transparent
way.
\item[(c)] Assumption (i) of (\ref{zalgl}) is crucial. Our result
does not cover the case $$W_{xy}(\sigma (x), \sigma(y)) = J_{xy}
\sigma(x)\sigma(y), \quad \sigma (x), \sigma(y)\in \mathds{R},$$
with random $J_{xy}$. For such models, only existence of Gibbs
fields on graphs satisfying (\ref{GraA}) has been established so
far, see \cite{KKP}.
\item[(d)] The i.i.d. condition in assumption (ii) of (\ref{zalgl}) is
typical for such situations \cite{BD,Berg,DKP,GM}.

\end{itemize}

\section{The Proof}

The proof will be based on the estimates obtained in Proposition
\ref{GGpn} and  Lemma \ref{Gra1lm}, see below, proved in the spirit
of similar statements in \cite{DKP,K}. For $\langle x, y\rangle \in
{\sf E}$ and $\beta
>0$, set
\begin{equation}
  \label{H9}
\varkappa_{xy}(\beta) =\exp\left( 4\beta \|W_{xy}\| \right)-1.
\end{equation}
For $\mathit{\Delta} \Subset {\sf V}$, by $\partial^{+}_{\sf
V}\mathit{\Delta}$ and $\partial^{-}_{\sf V}\mathit{\Delta}$ we
denote the outer and inner vertex boundaries of ${\mathit{\Delta}}$,
respectively.  That is,
\begin{gather*}
\partial^{+}_{\sf V}\mathit{\Delta} = \{y \in \varDelta^c: \exists x \in \varDelta \quad \langle x,y \rangle \in {\sf E}
\}, \\[.2cm] \partial^{-}_{\sf V}\mathit{\Delta} = \{x \in \varDelta: \exists y \in \varDelta^c \quad \langle x,y \rangle \in {\sf E}
\}.
\end{gather*}
For $z \in \varDelta\setminus \partial^{-}_{\sf V}\mathit{\Delta}$
and $x \in
\partial^{-}_{\sf V}\mathit{\Delta}$, define
\begin{equation}
  \label{H11}
Q^\beta_{\mathit{\Delta}} (z, x) = \sum_{\vartheta} \prod_{\langle
u,v \rangle \in {\sf E}_\vartheta} \varkappa_{uv}(\beta),
\end{equation}
where the sum is taken over all simple paths  in ${\mathit{\Delta}}$
whose origin and terminus are $z$ and $x$, respectively. Note that,
for fixed $\beta$, $\varDelta$ and $z,x$, $Q^\beta_{\mathit{\Delta}}
(x, y)$ is a random variable, the expected value of which can be
estimated by employing item (iii) of (\ref{zalgl}).

Along with the measures (\ref{3}), we consider also local measures
on $\mathcal{P}(S^{\mathit{\Delta}})$ of the following form
\begin{equation}
\label{miarani} \nu^\beta_{\mathit{\Delta}}
(d\sigma_\mathit{\Delta}|\xi)=\frac{1}{\widetilde{Z}^\beta_{\mathit{\Delta}}(\xi)}
\exp\left[ - \beta \widetilde{H}_{\mathit{\Delta}}
(\sigma_{\mathit{\Delta}}|\xi)\right]
\chi_{\mathit{\Delta}}(d\sigma_{\mathit{\Delta}}).
\end{equation}
where
\begin{eqnarray*}
\label{E1} \widetilde{H}_{\mathit{\Delta}}
(\sigma_{\mathit{\Delta}}|\xi)& = & -\frac{1}{2} \sum_{x,y\in
\mathit{\Delta} , \ x \sim y}
\left(W_{xy}(\sigma(x) , \sigma(y))+\|W_{xy}\|\right) \\[.2cm] \nonumber& - & \sum_{x\in
\varDelta , \ y \in \partial^{+}_{\sf V} \mathit{\Delta}, \ x \sim
y}W_{xy}(\sigma(x) , \xi(y)),
\end{eqnarray*}
and $\widetilde{Z}^\beta_{\mathit{\Delta}}(\xi)$ is the
corresponding normalization factor. Since $H$ defined in (\ref{2})
and $\widetilde{H}$ differ only on an additive  term independent of
$\sigma$, for each $\mathcal{F}_{\varDelta}$-measurable $f:S^{\sf
V}\to \mathds{R}$, it follows that
\begin{equation*}
  \int_{S^{\sf V}} f(\sigma) \pi^\beta_{\varDelta} (d \sigma|\xi) =
  \int_{S^{\varDelta}} h(\sigma_{\varDelta}) \nu^\beta_{\mathit{\Delta}}
(d\sigma_\mathit{\Delta}|\xi),
\end{equation*}
where $h:S^{\varDelta}\to \mathds{R}$ is such that $f(\sigma) =
h(\sigma_{\varDelta})$ for all $\sigma \in S^{\sf V}$. Given $z\in
{\sf V}$ and a continuous $h:S\to [0,1]$, we denote
\begin{equation}
\label{mag} M^\beta_{\mathit{\Delta}, z}(h|\xi) =
\int_{S^{\mathit{\Delta}}} h(\sigma(z)) \nu^\beta_\mathit{\Delta} (
d\sigma_{\mathit{\Delta}}|\xi).
\end{equation}
\begin{lemma}
  \label{Gra1lm}
For each $\mathit{\Delta} \Subset {\sf V}$, $z\in
\mathit{\Delta}\setminus \partial_{\sf V}^{-} \varDelta$ and
arbitrary $\xi, \eta \in S^{\mathsf{V}}$, (with probability one) the
following holds
\begin{equation} \label{27}
\left\vert  M^\beta_{\mathit{\Delta}, z}(h|\xi)-
M^\beta_{\mathit{\Delta}, z}(h|\eta)\right\vert \leq \sum_{x\in
\partial^{-}_{\sf V}\mathit{\Delta}}Q^\beta_{\mathit{\Delta}} (z, x).
\end{equation}
\end{lemma}
\begin{proof} Since both sides of (\ref{27}) are random, we are going to prove this inequality pointwise in
$\omega\in A$ with $P(A)=1$.
 By (\ref{miarani}) and (\ref{mag}) we obtain
\begin{eqnarray}\label{dowlm1}
M^\beta_{\mathit{\Delta}, z}(h|\xi)&-&M^\beta_{\mathit{\Delta},
z}(h|\eta)
\\[.2cm]
\nonumber &=&\int_{S^{\mathit{\Delta}}} \int_{S^{\mathit{\Delta}}}
\bigg{(}h(\sigma(z))
-h(\tilde{\sigma}(z))\bigg{)}\nu^\beta_\mathit{\Delta} (
d\sigma_{\mathit{\Delta}}|\xi)\nu^\beta_\mathit{\Delta} (
d\tilde{\sigma}_{\mathit{\Delta}}|\eta)\\ \nonumber &=&
\frac{1}{\widetilde{Z}^\beta_{\mathit{\Delta}}(\xi)\widetilde{Z}^\beta_{\mathit{\Delta}}(\eta)}\int_{S^{\mathit{\Delta}}}
\int_{S^{\mathit{\Delta}}} \bigg{(}h(\sigma(z))
-h(\tilde{\sigma}(z))\bigg{)}
 \\ \nonumber &\times & \prod_{\langle x,y \rangle\in {\sf
E}_{\mathit{\Delta}}} (1+\Gamma_{xy})\Psi_\mathit{\Delta}(\xi,
\eta)\chi_\mathit{\Delta} (
d\sigma_{\mathit{\Delta}})\chi_\mathit{\Delta} (
d\tilde{\sigma}_{\mathit{\Delta}})\\ \nonumber &=&
\frac{1}{\widetilde{Z}^\beta_{\mathit{\Delta}}(\xi)\widetilde{Z}^\beta_{\mathit{\Delta}}(\eta)}\sum_{{\sf
E}'\subset {\sf E}_{\mathit{\Delta}}}\int_{S^{\mathit{\Delta}}}
\int_{S^{\mathit{\Delta}}} \bigg{(}h(\sigma(z))
-h(\tilde{\sigma}(z))\bigg{)} \\
\nonumber &\times & \Gamma({\sf E}')\Psi_\mathit{\Delta}(\xi,
\eta)\chi_\mathit{\Delta} (
d\sigma_{\mathit{\Delta}})\chi_\mathit{\Delta} (
d\tilde{\sigma}_{\mathit{\Delta}}),\\ \nonumber
\end{eqnarray}
where
\begin{equation}
\label{Gamma} \Gamma({\sf E}')=\prod_{\langle x,y \rangle\in {\sf
E}'}\Gamma_{xy},
\end{equation}
\begin{eqnarray}\label{gammaxy}
\Gamma_{xy} & = &\exp\bigg{(}\beta\bigg{[}W_{xy}(\sigma(x) , \sigma(y))+\|W_{xy}\|\bigg{]}\\[.2cm] \nonumber
& + & \beta\bigg{[}W_{xy}(\tilde{\sigma}(x) , \tilde{\sigma}(y)) +
\|W_{xy}\|\bigg{]}\bigg{)}-1,
\end{eqnarray}
\begin{equation*}
\Psi_\mathit{\Delta}(\xi, \eta) = \prod_{x\in \mathit{\Delta} , \ y
\in \partial^{+}_{\sf V}\mathit{\Delta}, \ x\sim y}\exp\left[\beta
W_{xy}(\sigma(x),\xi(y)) +\beta
W_{xy}(\tilde{\sigma}(x),\eta(y))\right].
\end{equation*}
One observes that each $\Gamma_{xy}\geq 0$. To get this property of
$\Gamma_{xy}$ we added $\|W_{xy}\|$ to the corresponding interaction
terms. Fix some ${\sf E}'$ in the last line in (\ref{dowlm1}) and
consider the subgraph ${\sf G}'=(\mathit{\Delta},{\sf E}')$. If
$z\in \varDelta \setminus \partial^{-}_{\sf V}\mathit{\Delta}$ and
the boundary  $\partial^{-}_{\sf V}\mathit{\Delta}$ lie in distinct
connected components of ${\sf G}'$, then the integrals over the
spins $\sigma(z), \tilde{\sigma}(z)$ and over $\sigma(x),
\tilde{\sigma}(x)$ with $x\in\partial^{-}_{\sf V}\mathit{\Delta}$
get independent and hence the left-hand side of (\ref{dowlm1})
vanishes as the term $h(\sigma(x))- h(\tilde{\sigma}(x))$ is
antisymmetric with respect to the interchange $\sigma
\leftrightarrow \tilde{\sigma}$, whereas all $\Gamma_{xy}$ are
symmetric and the only break of this symmetry is related to the
fixed boundary spins $\xi(y)$ and $\eta(y)$, $y\in\partial^{+}_{\sf
V}\mathit{\Delta}$. Therefore, each non-vanishing term in
(\ref{dowlm1}) corresponds to a path $\vartheta(z,x)$ connecting $z$
to some $x\in\partial^{-}_{\sf V}\mathit{\Delta}$. Each such a path
$\vartheta(z,x)$ can be taken simple as every path contains a simple
subpath with the same origin and terminus. Let $\varTheta_\varDelta
(z)$ be the family of all simple paths in $\varDelta$ connecting $z$
and some $x\in
\partial^{-}_{\sf V} \varDelta$. Then the sum in (\ref{dowlm1}) can be restricted to those ${\sf
E}' \subset {\sf E}_{\mathit{\Delta}}$ that contain the edges of at
least one $\vartheta\in\varTheta_\varDelta (z)$. For such a path
$\vartheta$, set $\mathcal{E}_\vartheta=\{{\sf E}'\subset {\sf
E}_{\mathit{\Delta}}:{\sf E}_{\vartheta}\subset {\sf E}'\}$. Note
that, for distinct $\vartheta, \vartheta'\in
\varTheta_{\varDelta}(z)$, the corresponding families
$\mathcal{E}_{\vartheta}$ and $\mathcal{E}_{\vartheta'}$ need not be
disjoint as they may include those ${\sf E}'$ which contain both
${\sf E}_{\vartheta}$ and ${\sf E}_{\vartheta'}$. Set
\begin{equation}
\label{SA} \mathcal{E} = \bigcup_{\vartheta\in
\varTheta_{\varDelta}(z)} \mathcal{E}_{\vartheta}.
\end{equation}
That is, $\mathcal{E}$ includes all sets of edges ${\sf E}'\subset
{\sf E}_{\mathit{\Delta}}$ such that the corresponding graph ${\sf
G}' =(\varDelta, {\sf E}')$ has at least one subgraph ${\sf
G}_\vartheta$, $\vartheta \in \varTheta_{\varDelta}(z)$. Then
(\ref{dowlm1}) takes the form
\begin{eqnarray*}
M^\beta_{\mathit{\Delta}, z}(h|\xi)&-&M^\beta_{\mathit{\Delta},
z}(h|\eta) =
\frac{1}{\widetilde{Z}^\beta_{\mathit{\Delta}}(\xi)\widetilde{Z}^\beta_{\mathit{\Delta}}(\eta)}\\
\nonumber &\times&\sum_{{\sf
E}'\in\mathcal{E}}\int_{S^{\mathit{\Delta}}}
\int_{S^{\mathit{\Delta}}} (h(\sigma(z)) -h(\tilde{\sigma}(z)))\\
\nonumber
 &\times&\Gamma({\sf E}')\Psi_\mathit{\Delta}(\xi, \eta)
 \chi_\mathit{\Delta} (d\sigma_{\mathit{\Delta}})\chi_\mathit{\Delta} (d\tilde{\sigma}_{\mathit{\Delta}}).
\end{eqnarray*}
In view of the positivity of all $\Gamma_{xy}$ and the fact that
$h(\sigma)\in [0, 1]$, the latter yields
\begin{eqnarray}\label{dowlm2}
\nonumber |M^\beta_{\mathit{\Delta},
z}(h|\xi)&-&M^\beta_{\mathit{\Delta}, z}(h|\eta)| \leq
\frac{1}{\widetilde{Z}^\beta_{\mathit{\Delta}}(\xi)\widetilde{Z}^\beta_{\mathit{\Delta}}(\eta)}
\\[.2cm]
&\times&\sum_{{\sf E}'\in\mathcal{E}}\int_{S^{\mathit{\Delta}}}
\int_{S^{\mathit{\Delta}}} \Gamma({\sf E}')\Psi_\mathit{\Delta}(\xi,
\eta)\chi_\mathit{\Delta} (
d\sigma_{\mathit{\Delta}})\chi_\mathit{\Delta} (
d\tilde{\sigma}_{\mathit{\Delta}})\\ \nonumber &\leq&
\frac{1}{\widetilde{Z}^\beta_{\mathit{\Delta}}(\xi)\widetilde{Z}^\beta_{\mathit{\Delta}}(\eta)}
 \int_{S^{\mathit{\Delta}}} \int_{S^{\mathit{\Delta}}}\sum_{\vartheta\in \varTheta_{\varDelta}(z)}\Gamma({\sf E}_\vartheta)\sum_{{\sf E}'\in\mathcal{E}_{\vartheta}}\Gamma({\sf E}'\setminus {\sf E}_\vartheta) \\[.2cm] \nonumber
 &\times&\Psi_\mathit{\Delta}(\xi, \eta)\chi_\mathit{\Delta} ( d\sigma_{\mathit{\Delta}})
 \chi_\mathit{\Delta} (d\tilde{\sigma}_{\mathit{\Delta}}).\\ \nonumber
\end{eqnarray}
To get the second inequality in (\ref{dowlm2}) we used subadditivity
like $\sum_{x\in A\cup B} \psi (x)\leq \sum_{x\in A} \psi(x) +
\sum_{x\in B} \psi(x)$, $\psi(x)\geq 0$; which in our case is, see
(\ref{SA}),
\[
\sum_{{\sf E}' \in \mathcal{E}} \leq \sum_{\vartheta \in
\varTheta_{\varDelta}(z)} \sum_{{\sf E}'\in \mathcal{E}_\vartheta},
\]
as well as the fact that $\Gamma ({\sf E}') = \Gamma ({\sf
E}_\vartheta) \Gamma ({\sf E}'\setminus {\sf E}_\vartheta)$, see
(\ref{Gamma}). Note that
\begin{equation*}
0\leq \Gamma({\sf E}_\vartheta)\leq \prod_{\langle x,y\rangle\in{\sf
E}_\vartheta}\varkappa_{xy}(\beta),
\end{equation*}
which follows by (\ref{gammaxy}) and (\ref{H9}). We apply the latter
upper bound in (\ref{dowlm2}) and then obtain
\begin{eqnarray}
\label{dowlm3} |M^\beta_{\mathit{\Delta},
z}(h|\xi)&-&M^\beta_{\mathit{\Delta}, z}(h|\eta)| \leq
\frac{1}{\widetilde{Z}^\beta_{\mathit{\Delta}}(\xi)\widetilde{Z}^\beta_{\mathit{\Delta}}(\eta)}
\sum_{\vartheta\in \varTheta_{\varDelta}(z)}\prod_{\langle x,y\rangle\in{\sf E}_\vartheta}\varkappa_{xy}(\beta)\\[.2cm] \nonumber
&\times&\int_{S^{\mathit{\Delta}}} \int_{S^{\mathit{\Delta}}}
\prod_{\langle x,y\rangle\in{\sf E}_\vartheta}(1+\Gamma_{xy})\\[.2cm]
\nonumber &\times&\sum_{{\sf E}''\subset{\sf
E}_{\mathit{\Delta}}\setminus{\sf E}_{\vartheta}}\Gamma({\sf E}'')
\Psi_\mathit{\Delta}(\xi, \eta))\chi_\mathit{\Delta} (
d\sigma_{\mathit{\Delta}})\chi_\mathit{\Delta}
(d\tilde{\sigma}_{\mathit{\Delta}})\\ \nonumber &=& \sum_{\vartheta
\in \varTheta_{\varDelta}(z)} \prod_{\langle x,y\rangle\in{\sf
E}_\vartheta}\varkappa_{xy}(\beta),
\end{eqnarray}
which by (\ref{H11}) yields (\ref{27}). To get the latter equality
in (\ref{dowlm3}) we took into account that, see (\ref{dowlm1}),
\begin{eqnarray*}
&&\int_{S^{\mathit{\Delta}}} \int_{S^{\mathit{\Delta}}}
\prod_{\langle x,y\rangle\in{\sf
E}_\vartheta}(1+\Gamma_{xy})\sum_{{\sf E}''\subset{\sf
E}_{\mathit{\Delta}}\setminus{\sf E}_{\vartheta}}\Gamma({\sf E}'')\\
\nonumber &\times& \Psi_\mathit{\Delta}(\xi,
\eta)\chi_\mathit{\Delta}
(d\sigma_{\mathit{\Delta}})\chi_\mathit{\Delta} (
d\tilde{\sigma}_{\mathit{\Delta}})\\ \nonumber &=&
\int_{S^{\mathit{\Delta}}} \int_{S^{\mathit{\Delta}}} \prod_{\langle
x,y\rangle\in{\sf E}_\mathit{\Delta}}(1+\Gamma_{xy}) \\ \nonumber
&\times&\Psi_\mathit{\Delta}(\xi, \eta)\chi_\mathit{\Delta} (
d\sigma_{\mathit{\Delta}})\chi_\mathit{\Delta} (
d\tilde{\sigma}_{\mathit{\Delta}})=\widetilde{Z}^\beta_{\mathit{\Delta}}(\xi)\widetilde{Z}^\beta_{\mathit{\Delta}}(\eta),
\end{eqnarray*}
which completes the proof.
\end{proof}
\vskip.1cm
\begin{remark}
  \label{Rk}
A crucial aspect of the estimate in (\ref{27}) is that it is uniform
in the boundary spins $\xi$ and $\eta$. The main obstacle in proving
uniqueness for Gibbs fields with $W_{xy}$ not satisfying item (i) of
(\ref{zalgl}) is the lack of such uniformity.
\end{remark}

\noindent {\it Proof of Theorem \ref{2tm}.} By \cite[Theorem 1.33,
page 23]{Ge} the integrals as in (\ref{mag}) determine the
specification $\{\pi_{\varDelta}\}_{\varDelta \Subset {\sf V}}$;
hence, the uniqueness in question will follow by
\begin{equation*}
\inf_{\varDelta\Subset {\sf V}} \sup_{\xi, \eta \in S^{\sf
V}}|M^\beta_{\varDelta,z}(h|\xi) - M^\beta_{\varDelta,z}(h|\eta)| =
0,
\end{equation*}
holding (almost surely) for all $z\in {\sf V}$ and $h$ as in
(\ref{mag}). To prove this, we fix $z$ and $h$, and then pick a
cofinal sequence $\{\varDelta_l\}_{l\in \mathds{N}}$, for which we
show that the sequence of random variable
\begin{equation}
\label{Xl}
 X_l^\beta =  \sup_{\xi, \eta \in S^{\sf
V}}|M^\beta_{\varDelta_l,z}(h|\xi) -
M^\beta_{\varDelta_l,z}(h|\eta)|  \qquad l\in \mathds{N},
\end{equation}
converges to zero almost surely. Here \emph{cofinal} means that: (a)
$\varDelta_l \subset \varDelta_{l+1}$, $l\in \mathds{N}$; (b) each
$y\in {\sf V}$ lies in some $\varDelta_l$. For the mentioned $z$,
let $\{N_k\}_{k\in \mathds{N}}$ be the sequence as in Definition
\ref{KKdf}. Then we set
\begin{equation}
\label{Yk}
 Y_k^\beta =  \sup_{\xi, \eta \in S^{\sf
V}}|M^\beta_{{\sf V}_{N_k},z}(h|\xi) - M^\beta_{{\sf
V}_{N_k},z}(h|\eta)| \qquad k\in \mathds{N},
\end{equation}
where ${\sf V}_{r} = \{ x\in {\sf V}: \rho(z,x) \leq r\}$. By
condition (ii) of (\ref{zalgl}) and (\ref{dowlm3}) we get that
\begin{eqnarray}
\label{UQ} & & {\mathbb E}\left[Y^\beta_k\right] \leq
\sum_{\vartheta \in \varTheta_{{\sf V}_{N_k}}(z)} {\mathbb E}\left[
\prod_{\langle x,y\rangle\in{\sf
E}_\vartheta}\varkappa_{xy}(\beta)\right]\\[.2cm] \nonumber & & \quad
=\sum_{\vartheta \in \varTheta_{{\sf V}_{N_k}}(z)} \prod_{\langle
x,y\rangle\in{\sf E}_\vartheta}{\mathbb
E}\left[\varkappa_{xy}(\beta)\right] = \sum_{\vartheta \in
\varTheta_{{\sf V}_{N_k}}(z)} \varkappa^{\|\vartheta\|}(\beta),
\end{eqnarray}
with $\varkappa(\beta) := {\mathbb
E}\left[\varkappa_{xy}(\beta)\right]$. By condition (iii) of
(\ref{zalgl}) and Lebesgue's dominated convergence theorem it
follows that $\beta \mapsto \varkappa(\beta)$ is continuous and
increasing, and $\varkappa(\beta)\to 0$ as $\beta \to 0$. Set
$\beta_*$ such that $\varkappa(\beta_* ) = e^{-\gamma}$, where
$\gamma$ is as in (\ref{110}), see also the first part of
Proposition \ref{GGpn}. Thus, $\varkappa(\beta) e^\gamma < 1$
whenever $\beta < \beta_*$. Each $\vartheta \in \varTheta_{{\sf
V}_{N_k}}(z)$ is a simple path of length at least $N_k$, and hence
$\vartheta \in \varSigma_{N_k}^N(z)$ for some $N\geq N_k$. We take
this into account in (\ref{UQ}) and then obtain
\begin{gather}
  \label{UQ1}
{\mathbb E}\left[Y^\beta_k\right] \leq \sum_{N=N_k}^\infty
\sum_{\vartheta \in \varSigma_{N_k}^N (z)} \varkappa^{N}(\beta) =
\sum_{N=N_k}^\infty |\varSigma_{N_k}^N (z)| \varkappa^{N}(\beta)
\leq \frac{\left( \varkappa(\beta)e^\gamma\right)^{N_k}}{1 -
\varkappa(\beta)e^\gamma},
\end{gather}
where we used the estimate of $|\varSigma_{N_k}^N (z)|$ provided by
Proposition \ref{GGpn}. Since $N_k\to +\infty$, by (\ref{UQ1}) we
obtain that $Y^\beta_k\to 0$ in mean. Therefore, the sequence
$\{Y^\beta_k\}_{k\in \mathds{N}}$ contains a subsequence,
$\{Y^\beta_{k_l}\}_{l\in \mathds{N}}$, that converges to zero almost
surely, see, e.g., \cite[Theorem 3.4]{Gut}. Then we set $\varDelta_l
= {\sf }_{N_{k_l}}$, and hence  $X_l^\beta = Y^\beta_{k_l}$, see
(\ref{Xl}) and (\ref{Yk}) and obtain the convergence to be proved.



\section*{Acknowledgment}
The authors thank both referees whose valuable and important remarks
were essentially used in preparing the final version of the paper.
This work was supported in part by National Science Centre (NCN),
Poland, grant 2017/25/B/ST1/00051, which is cordially acknowledged.

\end{document}